\documentclass[journal]{IEEEtran}
\usepackage{amsmath,amsfonts}
\usepackage{algorithmic}
\usepackage{algorithm}
\usepackage{array}
\usepackage[caption=false,font=normalsize,labelfont=sf,textfont=sf]{subfig}
\usepackage{textcomp}
\usepackage{stfloats}
\usepackage{url}
\usepackage{verbatim}
\usepackage{graphicx}
\usepackage{cite}
\usepackage{lipsum}
\usepackage{hyperref}
\usepackage{tikz}
\usepackage{amssymb}
\usepackage{tablefootnote}
\usepackage{longtable}
\usepackage{listings}
\usepackage{scalerel}
\usepackage{graphicx}
\usepackage{cases}
\usepackage{lineno}

\usepackage{amsmath,amssymb,amsthm}
\usepackage{geometry}

\newtheorem{definition}{Definition}
\newtheorem{lemma}{Lemma}
\newtheorem{proposition}{Proposition}

\usepackage{hyperref}
\hypersetup{
    colorlinks=true, 
    linkcolor=blue,   
    citecolor=blue,   
    filecolor=blue,   %
    urlcolor=blue 
}

\newcommand{\orcidicon}{
    \scalerel*{
        \includegraphics{orcid.pdf}
    }{A}
}

\newcommand\orcid[1]{\href{https://orcid.org/#1}{\orcidicon}}

\begin{document}

\title{Nonlinear control and stability analysis of a unified Tethered UAV-winder system}


\author{Samuel O. Folorunsho \orcid{0000-0001-6386-9190},~\IEEEmembership{Member,~IEEE ,}
        Maggie~Ni, and~William R. Norris \orcid{0000-0002-4940-4458},~\IEEEmembership{Member,~IEEE}%


\IEEEcompsocitemizethanks{
    \IEEEcompsocthanksitem This research received no external funding.
    \IEEEcompsocthanksitem S.O Folorunsho and W.R Norris are with the Department
    of Industrial and Enterprise Systems Engineering, University of Illinois, Urbana-Champaign, IL 61801 (e-mails:~ sof3@illinois.edu,~wrnorris@illinois.edu). 
    \IEEEcompsocthanksitem M. Ni is with the Department of Aerospace Engineering, University of Illinois, Urbana-Champaign, IL 61801 (e-mail:~maggien2@illinois.edu).
    \IEEEcompsocthanksitem Manuscript received October 1, 2024.
}}%


\markboth{IEEE Transactions on Control Systems Technology,~Vol.~X, No.~Y, September, 2024}%
{Folorunsho \MakeLowercase{\textit{et al.}}: Nonlinear control and stability analysis of a unified Tethered UAV-winder system}

\maketitle

\begin{abstract}
This paper presents the development of a comprehensive dynamics and stabilizing control architecture for Tethered Unmanned Aerial Vehicle (TUAV) systems. The proposed architecture integrates both onboard and ground-based controllers, employing nonlinear backstepping control techniques to achieve asymptotic stability of the TUAV's equilibrium. The onboard controllers are responsible for the position and attitude control of the TUAV, while the ground controllers regulate the winder mechanism to maintain the desired tether length, ensuring it retains its catenary form. Simulation results demonstrate the ability of the TUAV system to accurately track linear and circular trajectories, ensuring robust performance under various operational scenarios. The code and movies demonstrating the performance of the system can be found at \href{https://github.com/sof-danny/TUAV\_system\_control}{https://github.com/sof-danny/TUAV\_system\_control.}

\end{abstract}


\section{Introduction}
\IEEEPARstart{U}{nmanned} Aerial Vehicles (UAVs) have a wide range of applications across a diverse field of industries and scientific disciplines; some notable ones include agriculture \cite{Patel_2013, Hasan_2020}, environmental monitoring \cite{Alvear_2017}, and surveillance \cite{Siebert_2014, Sampedro_2019}. However, their versatility is hindered by their insufficient flight times \cite{folorunsho2024redefiningaerialinnovationautonomous}. This limitation is particularly significant as the applicability of UAVs can be extended given prolonged periods of operation to collect meaningful data or maintain continuous observation. 

Tethered Unmanned Aerial Vehicles (TUAVs) are a viable and popular alternative to conventional UAVs due to their unique ability to remain airborne for extended periods while connected to a ground power station. This tethering provides a constant power supply, eliminating the need for frequent battery changes and allowing for real-time data transmission. TUAVs have been successfully applied in fields such as navigation \cite{Papachristos_2014, Miki_2019, Xiao_2018}, surveillance \cite{Dufek_2017, Xiao_2017}, meteorology ~\cite{Rico_2021, Qi_2019, Carrozzo_2019, Rico_2020}, and telecommunications \cite{Saif_2021, Safwat_2022, Vishnevsky_2019}. Studies have demonstrated their effectiveness in maintaining a stable presence for extended periods, which is particularly valuable for tasks such as persistent aerial monitoring and high-altitude communication. Research has explored innovative use cases of TUAVs such as providing backhaul connectivity in post-disaster scenarios \cite{Saif_2021} and their effectiveness in prolonged atmospheric data collection \cite{Rico_2021}. TUAVs have also been used in collaboration with Unmanned Ground Vehicles (UGVs) for assisted mapping and navigation in challenging environments ~\cite{Papachristos_2014, Miki_2019}. The concept has been extended to water vehicles with TUAVs monitoring aquatic environments for drowning victims \cite{Dufek_2017}, manipulating an ocean buoy \cite{Kourani_2023}, and detecting oil slicks ~\cite{Muttin_2011}.

With the increasing deployment of TUAVs, a significant area of research has focused on the stability and control of these systems, particularly due to the complex dynamics introduced by the tether, especially when it assumes a non-taut catenary configuration. In this paper, a comprehensive control system is proposed that includes: (1) a three-dimensional dynamical model of the UAV, (2) a unified mathematical representation of the TUAV-winder system with the tether modeled as a catenary, and (3) both aerial and ground-based controllers utilizing nonlinear backstepping techniques. The onboard controllers managed the TUAV's position and attitude, while the ground controllers regulated the winder mechanism to maintain the desired tether length, ensuring it retains its catenary form. Using the developed controllers, asymptotic stability is achieved, ensuring both set point regulation and trajectory tracking. Simulation results are provided to validate the effectiveness of the proposed control architecture.

\section{Related Works}
Designing control systems for Tethered Unmanned Aerial Vehicles (TUAVs) requires addressing several critical challenges, including maintaining system stability in the presence of complex tether dynamics, ensuring robust performance under varying wind conditions, and adhering to safety constraints imposed by the tether. The control algorithms and strategies must also be tailored to the specific operational requirements of the TUAV.

A significant body of literature has explored these challenges through the development of various control models. However, many of these studies simplify the tether as a straight, massless, and taut line, which does not accurately capture the full dynamics of a tethered system \cite{Nicotra_2017, Kourani_2023, Glick_2018}. For instance, a unified model of the integrated ground-aerial system has been proposed, utilizing a Model Predictive Control (MPC) algorithm to achieve optimal control under state and input constraints~\cite{Glick_2018}. Similarly, an MPC scheme based on velocity kinematics and a constant curvature assumption has been presented \cite{Chien_2021}. In another study, a backtracking Reference Governor (RG) algorithm was employed~\cite{Nicotra_2017}. Notably, these models are restricted to planar (2D) representations of the TUAV, which limits their ability to capture the full spatial dynamics of the system.

A feedforward and Proportional-Derivative (PD) position controller was implemented to address control challenges in precision agriculture applications, \cite{Glick_2016}. Proportional-Integral-Derivative (PID) controllers are frequently utilized in the literature due to their simplicity, effectiveness, and robustness \cite{Dufek_2017, Dicembrini_2020, Abantas_2022, Xiao_2018, Xiao_2019}. However, PID controllers have notable limitations in managing the complexities of dynamic systems, such as a lack of adaptability, the potential for oscillatory responses \cite{Chien_2021}, and difficulties in tuning ~\cite{Abantas_2022}.

Given the multi-component nature of a TUAV system, which includes the UAV, the tether, and the winder, there is a need for a control system that can accommodate its inherent complexity. Linear controllers like the one employed in \cite{liu2021dynamic} often underperform in scenarios involving wide operational ranges or significant model uncertainties \cite{Kokotovic_1991, Slotine_1991}. Despite this, the application of nonlinear control methods to TUAVs remains relatively uncommon. This oversimplification is a common issue in many publications that reduce the mathematical model to 2D or fail to incorporate the winder, tether, and UAV as an integrated system.

This research aims to address these limitations by applying nonlinear backstepping control, a method particularly well-suited for handling the complex dynamics and nonlinearities characteristic of TUAV operations \cite{Madani_2006, Al-Younes_2010}.

\section{Preliminaries}
Let \( \mathbb{R}^{>0} \) denote the set \(\{x \in \mathbb{R} : x > 0\}\), and let \( \mathbb{R}^{\geq 0} \) denote the set \(\{x \in \mathbb{R} : x \geq 0\}\). Also, let \(S(\cdot)\), \(C(\cdot)\), and \(T(\cdot)\) represent the sine, cosine, and tangent functions for some angle \((\cdot)\) respectively. In addition, let \(\|\cdot\|\) denote the \(L_2\) norm.

A system's equilibrium is said to be asymptotically stable if, for every \(\epsilon \in \mathbb{R}^{>0}\), there exists a \(\delta \in \mathbb{R}^{>0}\) such that if the initial condition satisfies \(\|x(0)\| < \delta\), then \(\|x(t)\| < \epsilon\) for all \(t \geq 0\), and \(\|x(t)\| \to 0\) as \(t \to \infty\).

\begin{figure}[ht]
\centering
\includegraphics[width=0.5\textwidth]{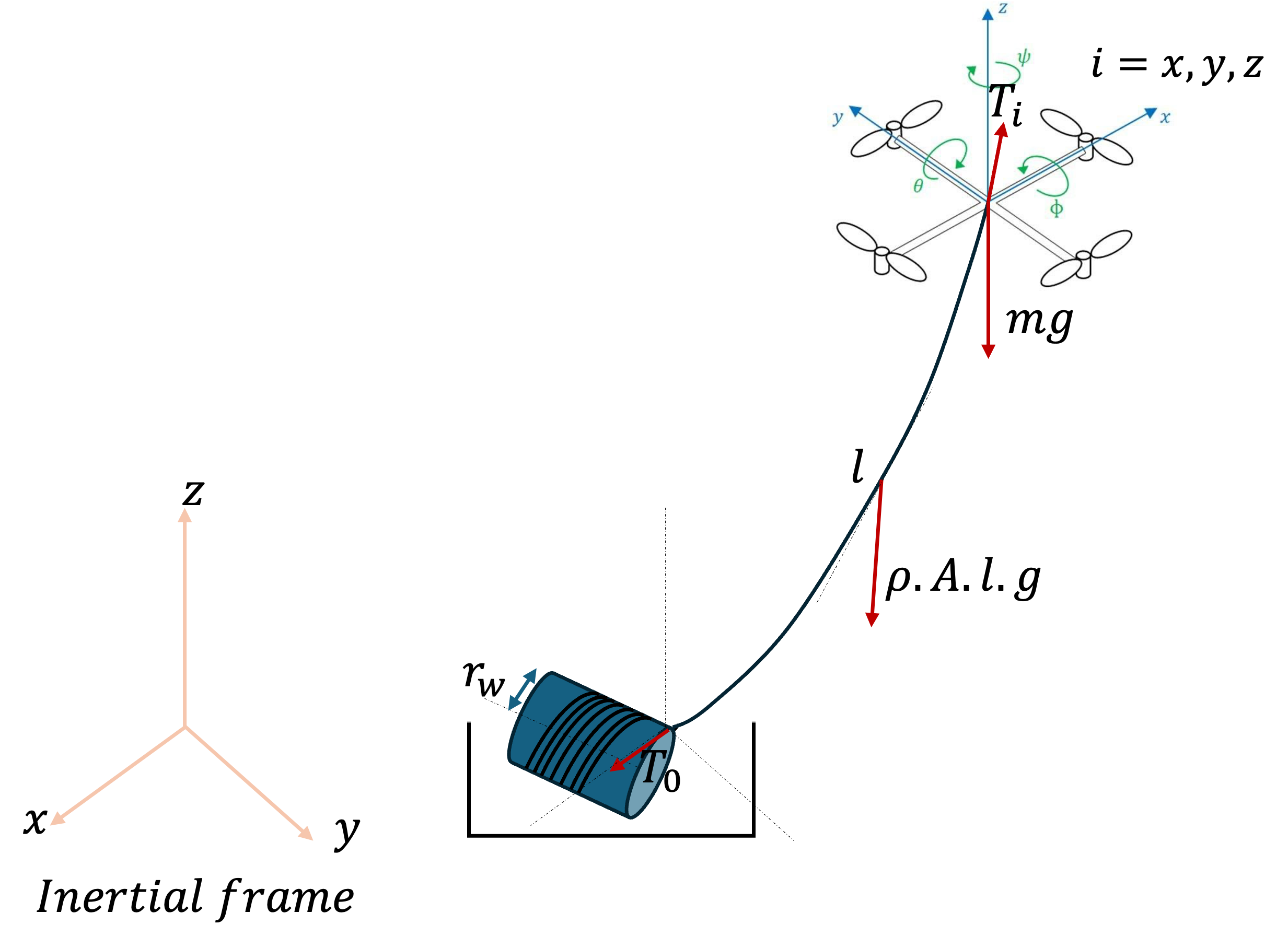}
\caption{The tethered drone system with winder}
\captionsetup{justification=centering}
\label{tuavsys}
\end{figure}

\section{System Description}
\label{sysdes}
Consider the system represented by Fig. \ref{tuavsys} of an Unmanned Aerial Vehicle (UAV) of mass \( m \in \mathbb{R}^{>0} \) and moments of inertia \(I \in \mathbb{R}_{>0}^3\)
, physically connected to a winder of radius \( r_w \in \mathbb{R}^{>0} \) located at a ground station in an inertial frame via a tether of density \( \rho \in \mathbb{R}^{>0} \) in a catenary form. Two frames of reference are considered: the inertial frame whose origin is located at the pivot point where the tether is attached to the winder, with the third axis pointing upward, \( \mathbf{g} \), and the body-fixed frame whose origin is located at the mass center of the UAV. The UAV is acted upon by its own weight, drag forces in the \( x \), \( y \), and \( z \) directions, the tension in the cable,\( \mathbf{T} \)  and the weight of the tether depending on the material of the cable.

\subsection{The tether model}
Consider a tether in a catenary form governed by the equation:
\[ z(x) = a \cosh\left(\frac{x}{a}\right), \quad a \in \mathbb{R}^{>0} \]
where \(a\) is a parameter related to the length of the cable and the distance between the points it is attached to. The catenary parameter \(a\) is given by:
\[ a = \frac{T_0}{\rho g}, \quad T_0, \rho, g \in \mathbb{R}^{>0} \]
where \( T_0 \) is the initial horizontal tension, \( \rho \) is the linear density of the tether, and \( g \) is the gravitational acceleration.

For a section of the tethered cable with horizontal coordinates from \(0\) to \(x\), the tension at the top of the tether \(T_1\) can be described as in equation \ref{ten1} \cite{liu2021dynamic}. 

\begin{equation}
\label{ten1}
\begin{aligned}
    T_1 &= T_0 + \rho A z g, \quad T_1 , A, z \in \mathbb{R}^{>0}
\end{aligned}
\end{equation}

where \( A \) is the cross-sectional area of the tether, and \( z \) is the vertical component of the tether's position.

The length of the tether at any point can be described by the transcendental equation \cite{talke2018catenary}:

\begin{equation}
\label{eqn2}
L = \sqrt{2a\sinh\left(\frac{x}{2a}\right) + z^2}, \quad L \in \mathbb{R}^{>0}
\end{equation}

Assuming the pulling force of the tether on the center of mass of the UAV at a hovering moment is \( T_1 \), this force can be resolved into its \( x \), \( y \), and \( z \) components as:
\[ \mathbf{T} = \begin{bmatrix} T_X & T_Y & T_Z \end{bmatrix}^T, \quad \mathbf{T} \in \mathbb{R}^3 \]

where the components are given by:
\[ 
T_X = T_1 \cos \alpha \sin \beta \]
\[T_Y = T_1 \cos \alpha \cos \beta \]
\[T_Z = T_1 \sin \alpha \]

with:
\[ 
\beta = \arctan\left(\sinh\left(\frac{x - x_0}{a}\right)\right) \]

\[\alpha = \arctan\left(\sinh\left(\frac{-x_0}{a}\right)\right) \]

\subsection{The winder dynamics}
The tether is connected at one end to a winder. Let \( \vartheta \in \mathbb{R} \) and \( \dot{\vartheta} \in \mathbb{R} \) represent the angular position and velocity of the winch, respectively. The mass of the winch without the tether is \( m_w \in \mathbb{R}_{>0} \), and the mass of the tether per unit length is \( \rho \in \mathbb{R}_{>0} \). The maximum tether length is \( L_{T} \in \mathbb{R}_{>0} \), and the effective radius of the winch is \( r_w \in \mathbb{R}_{>0} \). The inner radius of the winch drum is \( r_i \in \mathbb{R}_{>0} \). The viscous friction coefficient of the winch is \( \beta_w \in \mathbb{R}_{>0} \), and the winch torque (control input) is \( U_{win} \in \mathbb{R} \). The tether stiffness constant is \( K_t \in \mathbb{R}_{>0} \). The position of the UAV is \( \mathbf{P}_{1} \in \mathbb{R}^3 \), and the position of the ground station (winch base) is \( \mathbf{P}_{0} = (0, 0, 0) \).

The mass of the winch, including the unreleased tether, is computed as:
\begin{equation}
m_w = \bar{m}_w + \left( L_{T} - r_w \vartheta \right) \rho
\end{equation}

where $\bar{m}_w$ is the mass of the winch without the tether.

Considering a hollow drum with inner radius \( r_i \), the moment of inertia of the winch is approximated as:

\begin{equation}
I_w = \frac{1}{2} m_w \left( r_w^2 + r_i^2 \right)
\end{equation}

The pulling force vector \( \mathbf{F}_{p} \) exerted by the tether is computed on the basis of its stiffness \( K_t \) and its elongation \( e_{t} \), computed as:
\[
e_{t} = \max \left( 0, \| \mathbf{P}_{1} - \mathbf{P}_{0} \| - r_w \vartheta \right)
\]
\[
\mathbf{F}_{p} = K_t e_{t} \frac{\mathbf{P}_{1} - \mathbf{P}_{0}}{\| \mathbf{P}_{1} - \mathbf{P}_{0} \|}
\]

From the equilibrium of moments around the rotational axis, the equation of motion for the winch is obtained as in \cite{fagiano2017systems}:

\begin{equation}
\ddot{\vartheta} = \frac{1}{I_w} \left( r_e \| \mathbf{F}_{p} \| - \beta_w \dot{\vartheta}(t) + r_w U_{win} \right)
\end{equation}

\textbf{Assumption 1: The cable is inelastic.}

Given this assumption, the pulling force \( \mathbf{F}_{p} \) exerted by the tether is zero. Thus:
\[
\mathbf{F}_{p} = 0
\]

By this assumption, the equation of motion for the winch reduces to:

\begin{equation}
\ddot{\vartheta} = \frac{1}{I_w} \left( - \beta_w \dot{\vartheta} + r_w U_{win} \right)
\end{equation}

\subsection{The UAV model}
Consider the position vector of the center of mass of the UAV defined by $\boldsymbol{\xi} = (x,y,z) \in \mathbb{R}^3$ and the UAV’s Euler angles (the orientation) are expressed by $\boldsymbol{\eta} = (\psi, \theta, \phi) \in \mathbb{R}^3$, where $\psi$ is the yaw angle around the z-axis, $\theta$ is the pitch angle around the y-axis, and $\phi$ is the roll angle around the x-axis. Let the generalized coordinates of the UAV then be expressed by $\boldsymbol{q} = (\boldsymbol{\xi}, \boldsymbol{\eta}) \in \mathbb{R}^6$.

The translational kinetic energy,

\begin{equation}
T_{\text{trans}} = \frac{1}{2}m \dot{\boldsymbol{\xi}}^T \dot{\boldsymbol{\eta}},
\end{equation}
while the rotational kinetic energy,
\begin{equation}
T_{\text{rot}} = \frac{1}{2}\boldsymbol{\omega}^T \mathbf{I} \boldsymbol{\omega},
\end{equation}
where $\boldsymbol{\omega}$ is the vector of the angular velocity and $\mathbf{I}$ is the inertia matrix. The potential energy of the UAV,
\begin{equation}
\mathcal{P} = mgz,
\end{equation}
with $g$ as acceleration due to gravity.

Therefore, the Lagrangian, with $\mathcal{T}$ as the sum of $T_\text{rot}$ and $T_\text{trans}$ is:
\[
\mathcal{L} = \mathcal{T} - \mathcal{P}.
\]

The angular velocity vector $\boldsymbol{\omega}$ resolved in the body-fixed frame is related to the generalized velocities $\dot{\boldsymbol{\eta}}$ through this relationship:
\begin{equation}
\boldsymbol{\omega} = \mathbf{W}_{\boldsymbol{\eta}} \dot{\boldsymbol{\eta}},
\end{equation}
where
\[
\mathbf{W}_{\boldsymbol{\eta}} = \begin{bmatrix}
-\sin \theta & 0 & 1 \\
\cos \theta \sin \phi & \cos \phi & 0 \\
\cos \theta \cos \phi & -\sin \phi & 0
\end{bmatrix}.
\]

Hence,
\begin{equation}
T_{\text{rot}} = \frac{1}{2} \dot{\boldsymbol{\eta}}^T \mathbf{J} \dot{\boldsymbol{\eta}},
\end{equation}
with
\[
\mathbf{J} = \mathbf{W}_{\boldsymbol{\eta}}^T \mathbf{I} \mathbf{W}_{\boldsymbol{\eta}},
\]
and
\[
\mathbf{I} = \begin{bmatrix}
I_{xx} & 0 & 0 \\
0 & I_{yy} & 0 \\
0 & 0 & I_{zz}
\end{bmatrix}.
\]

The matrix $\mathbf{J} = \mathbf{J}(\boldsymbol{\eta})$ is therefore expressed directly in terms of the generalized coordinates $\boldsymbol{\eta}$.

The Euler-Lagrange equation with external forces is given by:
\[
\frac{d}{dt} \left( \frac{\partial L}{\partial \dot{\boldsymbol{q}}} \right) - \frac{\partial L}{\partial \boldsymbol{q}} = \begin{bmatrix} F_i \ (i = x, y, z) \\ \boldsymbol{\tau} \end{bmatrix},
\]
where $F_i = \mathbf{R} \mathbf{f} \in \mathbb{R}^3$ represent the translational forces due to the main thrust including the tether tension. $\boldsymbol{\tau} \in \mathbb{R}^3$ is the yaw, pitch, and roll moments, and $\mathbf{R}$ is the rotation matrix from the body frame to the inertial frame, defined by:
\[
\mathbf{R} = \begin{bmatrix}
C_\theta C_\psi & C_\psi S_\theta S_\phi - C_\phi S_\psi & S_\phi S_\psi + C_\phi C_\psi S_\theta \\
C_\psi & C_\phi C_\psi + S_\theta S_\phi S_\psi & C_\phi S_\theta S_\psi - C_\psi S_\phi \\
-S_\theta & C_\theta S_\phi & C_\theta C_\phi
\end{bmatrix}.
\]

Based on this, the dynamics of the UAV , with the tether and the winder are obtained \cite{carrillo2012quad, liu2021dynamic} as in equation \ref{sys}:

\begin{equation}
\label{sys}
    \begin{cases}
        \ddot{x} = \frac{U_f (C_\psi C_\phi S_\theta + S_\psi S_\phi) + mrv - mqw - mgS_\theta - T_1 C_\alpha S_\beta - A_{x}u}{m} \\
        \ddot{y} = \frac{U_f (C_\phi S_\psi S_\theta - C_\psi S_\phi) + mru - mpw - mgC_\theta S_\phi + T_1 C_\alpha - A_{y}v}{m} \\
        \ddot{z} = \frac{U_f C_\theta C_\phi + mqu - mpv + mgC_\theta C_\phi - T_1 S_\alpha - A_{z}w}{m} \\
        \ddot{\phi} = \frac{U_\phi - qr(I_{yy} - I_{zz})}{I_{xx}} \\
        \ddot{\theta} = \frac{U_\theta + pr(I_{xx} - I_{zz})}{I_{yy}} \\
        \ddot{\psi} = \frac{U_\psi - pq(I_{xx} - I_{yy})}{I_{zz}} \\
        \ddot{\vartheta} = \frac{1}{I_w} \left( - \beta_w \dot{\vartheta} + r_w U_{win} \right)
    \end{cases}
\end{equation}

Where \( u = \dot{x} \), \( v = \dot{y} \), \( w = \dot{z} \), \( p = \dot{\phi} \), \( q = \dot{\theta} \), \( r = \dot{\psi} \) are the time derivatives of the position and the Euler angles. The drag acting in the \( x, y, z \) directions is represented as \( A_x, A_y, A_z \). Additionally, \( U_f \) is denoted as the lift input, \( U_\phi \), the roll input, \( U_\theta \), the pitch input, and \( U_\psi \) as the yaw input.

\subsection{Control Objectives}
The primary objective of this paper is to formulate and implement a nonlinear control strategy for a tethered Unmanned Aerial Vehicle (UAV) system, where the tether is modeled as a catenary (non-taut) configuration and includes dynamic effects from a winder mechanism. To achieve this, a nonlinear backstepping control methodology is employed, tailored to accommodate the complexities of the system's nonlinear dynamics. The proposed control scheme is evaluated through numerical simulations for trajectory tracking and set-point regulation tasks, ensuring robust performance across various operating conditions.

\section{Controller design and stability analysis}
\label{control}
In this section, the backstepping controller for both the onboard and the ground control is developed. The onboard controllers are responsible for the position and attitude control of the TUAV, while the ground control is responsible for the winder and maintaining the catenary shape of the tether. The overall control architecture is shown in Fig.~\ref{control_arch}.

\begin{figure*}[ht]
\centering
\includegraphics[width=0.6\textwidth]{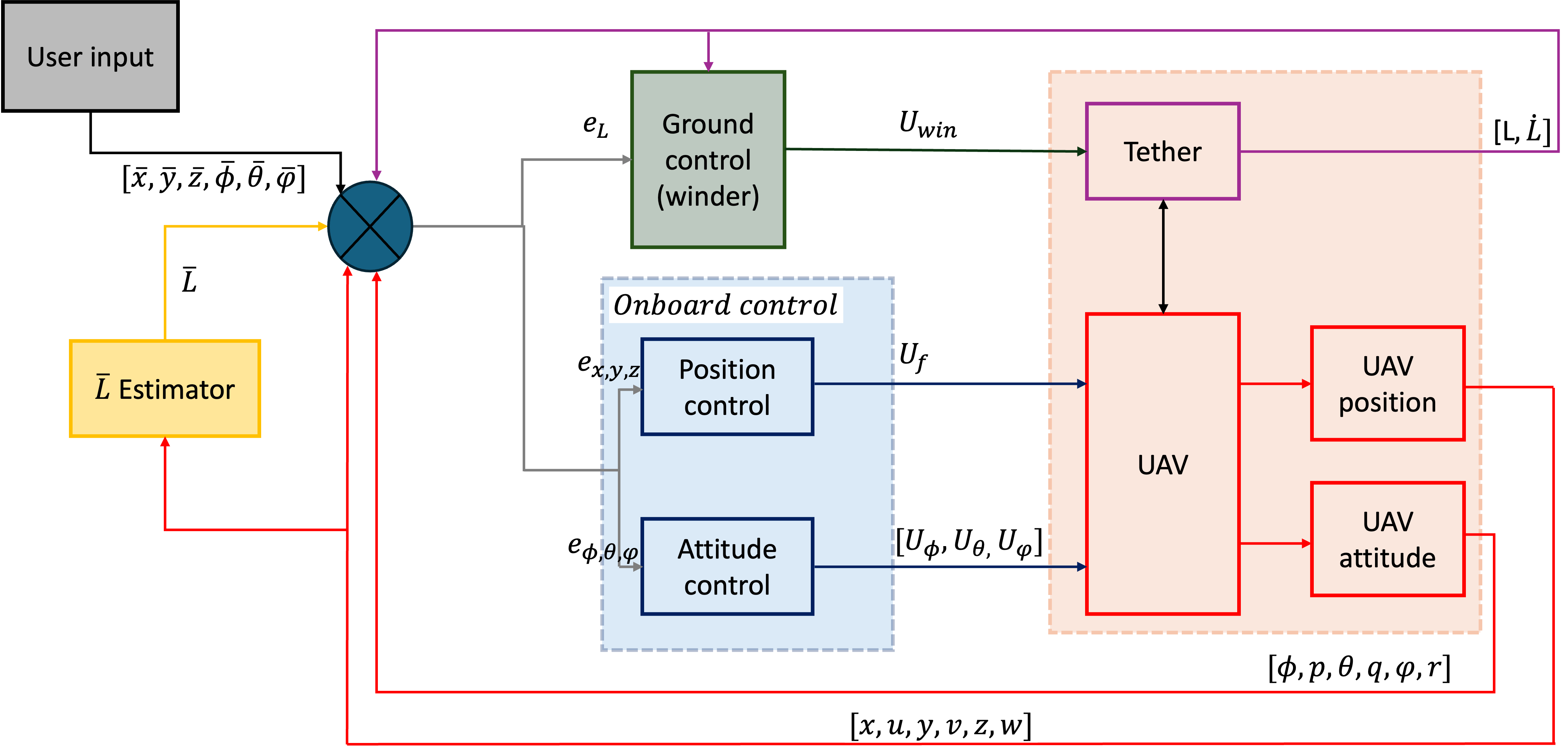}
\caption{Overall control system architecture}
\label{control_arch}
\end{figure*}

The \(\bar{L}\) Estimator block computes the desired tether length \(\bar{L} \in \mathbb{R}^{\geq 0}\) based on the current position of the UAV \((x, y, z) \in \mathbb{R}^3\). This calculation ensures that the tether length corresponds accurately to the TUAV's current position, maintaining the desired catenary shape. The user inputs the desired position \((\bar{x}, \bar{y}, \bar{z}) \in \mathbb{R}^3\) and desired attitude \((\bar{\phi}, \bar{\theta}, \bar{\psi}) \in \mathbb{R}^3\). The summing junction receives these inputs and the current position \((x, y, z)\), and the desired tether length \(\bar{L}\). It computes the errors in position, attitude, and tether length. Let \(\mathbf{e_{x,y,z}} \in \mathbb{R}^3\) represent the position error, \(\mathbf{e_{\phi,\theta,\psi}} \in \mathbb{R}^3\) represent the attitude error, and \(e_L \in \mathbb{R}\) represent the tether length error. These errors are then fed into the respective controllers to adjust the UAV's behavior accordingly.

The Ground Control, or winder, receives the tether length error \(e_L\) from the Reference Generator and adjusts the tether length \(L \in \mathbb{R}^{\geq 0}\) to minimize this error, ensuring that the tether maintains the appropriate catenary shape as the UAV moves. The Onboard Control system is divided into two main controllers: Position Control and Attitude Control. Position Control receives the current position \((x, y, z)\) of the UAV and sends commands to adjust the UAV's position to match the desired trajectory \((\bar{x}, \bar{y}, \bar{z})\), ensuring the UAV moves accurately to the target coordinates. Attitude Control receives the current attitude \((\phi, \theta, \psi) \in \mathbb{R}^3\) of the UAV and sends commands to adjust its attitude, maintaining stability and proper orientation during flight.

\subsection{Altitude controller}
First, the focus is on the altitude control of the system described by equation \ref{sys}.

The \( z \) position of the UAV can be described in state-space form. Let \( x_5 \in \mathbb{R}_{\geq 0} \), with \( 0 \leq x_5 \leq L_T \), represent the state corresponding to the \( z \) position of the UAV, and let \( x_6 \) denote the time derivative of this state. The resulting state-space equations can then be described as follows:

\begin{equation}\label{mainsys1}
\begin{aligned}
\dot{x}_5 &= x_6 \\
\dot{x}_6 &= \frac{1}{m}\left( U_f \cos x_9 \cos x_7 +
m x_{10} x_{12} - m x_8 x_4 \right) \\
&\quad + \frac{1}{m}\left( mg \cos x_9 \cos x_7 + T_z - A_z x_6 \right)
\end{aligned}
\end{equation}

A virtual controller for the system is designed as a starting point.

\begin{definition}
Let \( x_6 \) be the input to this system, and define \( x_6 = \phi(x_5) \) such that \( \phi(0) = 0 \). This choice of \( \phi \) ensures that the origin of the system is an equilibrium point and guarantees asymptotic stability (AS).

\end{definition}

\begin{proposition}
The choice of $\phi(x_5) = -k_1 x_5$ is made.
\end{proposition}

\begin{proof}
\label{lyapunovC}
Choosing a Lyapunov candidate, \( V_{1} : \mathbb{R}^n \to \mathbb{R} \), where \( V(x) > 0 \) for all \( x \neq 0 \) and \( V(0) = 0 \). 
\begin{equation}
\label{ly1}
V_1 = \frac{1}{2} x_5^2.
\end{equation}

\[
\dot{V}_1 = x_5 \dot{x}_5 = x_5 \cdot (-k_1 x_5) = -k_1 x_5^2 < 0 
\]

This shows asymptotic stability. Hence, the chosen controller \( -k_1 x_5 \) will stabilize the system.
\end{proof}

Now, consider a transformation, \( z_1 = x_6 - \phi(x_5) \).

\[
x_6 = z_1 + \phi(x_5)
\]

The time derivative of this transformation is given by:

\[
\dot{z}_1 = \dot{x}_6 - \dot{\phi}(x_5)
\]

For simplicity, let \(\varsigma = \dot{x}_6 \) and \(\varsigma - \dot{\phi}(x_5) = \varkappa\)

Then,

\begin{equation*}
\label{tr1}
\dot{z}_1 = \varsigma - \dot{\phi}(x_5) = \varkappa
\end{equation*}

\begin{lemma}
\label{lemma1}
The asymptotic stability of the transformed system, with the state variables defined as in Equation \ref{ts1}, can be established using a composite Lyapunov function that follows the definition in \ref{lyapunovC}. This Lyapunov function \( V(x) \) is constructed by combining the original Lyapunov function in \ref{ly1} with an additional Lyapunov function that incorporates the state transformation. Specifically, \( V(x) \) is formulated for system in equation \ref{ts1} as follows:
\end{lemma}

\begin{equation}
\label{ts1}
\begin{aligned}
\dot{x}_5 &= z_1 + \phi(x_5) \\
\dot{z}_1 &= \varkappa
\end{aligned}
\end{equation}

\begin{proof}
Using the composite Lyapunov function
\[
V_{c1} = V_1 + \frac{1}{2} z_1^2
\]

\[
\dot{V}_{c1} = \dot{V}_1 + z_1 \dot{z}_1 = x_5z_1 -k_1 x_5^2 + z_1 (\varkappa)
\]

To ensure the asymptotic stability of this system, the following choice is made: 

\[
\varkappa = -x_5 - k_2 z_1
\]

Hence,

\[
\dot{V}_{c1} = -k_1 x_5^2 + z_1 (x_5 - x_5 - k_2 z_1) = -k_1 x_5^2 - k_2 z_1^2 < 0
\]

which shows asymptotic stability (AS) with \( K_1 \) and \( K_2 \) as controller gains, where \( K_1 \geq 0 \) and \( K_2 \geq 0 \).

\end{proof}

Now, to find \( U_f \),the following fact can be used: 

\( \varkappa = \varsigma - \dot{\phi}(x_5) \).

\begin{equation}
\begin{aligned}
\varsigma &= \varkappa + \dot{\phi}(x_5)
\end{aligned}
\end{equation}

For the system described by equation \eqref{mainsys1}, the control input \( U_f \) given in \eqref{control1} will stabilize the system:

\begin{equation*}
\begin{aligned}
\frac{1}{m}\left( U_f \cos x_9 \cos x_7 + m x_{10} x_{12} - m x_8 x_4 \right) & \nonumber \\
+ \frac{1}{m}\left( mg \cos x_9 \cos x_7 + T_z - A_z x_6 \right) & \nonumber \\
= -x_5 - k_2 z_1 - k_1 &
\end{aligned}
\end{equation*}

\begin{equation}
\label{control1}
\begin{aligned}
U_f &= \frac{-m x_{10} x_{12} + m x_8 x_4 - mg \cos(x_9) \cos(x_7) - T_z}{\cos(x_9) \cos(x_7)} \\
&\quad + \frac{A_z x_6 - m x_5 - m k_2 z_1 - m k_1}{\cos(x_9) \cos(x_7)}
\end{aligned}
\end{equation}

A modification of the form,
\[
e_z = x_5 - \bar{x_5}
\]
where \( e_z \) is the error in the z direction and \(\bar{x_5}\) is the desired trajectory, will ensure the controller can be used for tracking purposes.

A modified version of this controller will be used to control the x and y directions of the TUAV. 

\subsection{Attitude controller}

This section focuses on developing a backstepping controller for the attitude of the TUAV. There are three controller inputs, but the process of developing them is similar. Therefore, the focus will be on the steps for the $\tau_\phi$, followed by the presentation of the developed controller inputs for the $\tau_\theta$ and $\tau_\psi$.

The rotational angle \(\phi\) can be described by the state equations below:

\begin{equation}
\label{mainsys2}
\begin{aligned}
\dot{x}_7 &= x_8 \\
\dot{x}_8 &= \frac{U_{\phi} - x_{10} x_{12} I_{yy} + x_{10} x_{12} I_{zz}}{I_{xx}}
\end{aligned}
\end{equation}

First, a virtual controller is designed: 

\begin{definition}
Let \( x_8 \) be the input to this system, and define \( x_8 = \phi(x_7) \) such that \( \phi(0) = 0 \). This choice of \( \phi \) ensures that the origin of the system is an equilibrium point and guarantees asymptotic stability (AS).
\end{definition}

\begin{proposition}
The choice of \( \phi(x_7) = -k_3 x_7 \) is made.
\end{proposition}

\begin{proof}
Consider a positive definite and radially unbounded Lyapunov candidate as in \ref{lyapunovC}, \( V_2 = \frac{1}{2} x_7^2 \).

\[
\dot{V}_2 = x_7 \dot{x}_7 = x_7 \cdot (-k_3 x_7) = -k_3 x_7^2 < 0 
\]

This shows asymptotic stability. Hence, the chosen controller \( -k_3 x_7 \) will stabilize the system.
\end{proof}

Now, consider a transformation, \( z_2 = x_8 - \phi(x_7) \).

\[
x_8 = z_2 + \phi(x_7)
\]

The time derivative of this transformation is given by:

\[
\dot{z}_2 = \dot{x}_8 - \dot{\phi}(x_7)
\]

Let \(\varsigma_2 = \dot{x}_8 \) and \(\varsigma_2 - \dot{\phi}(x_7) = \varkappa_2\)

Then,

\[
\dot{z}_2 = \varsigma_2 - \dot{\phi}(x_7) = \varkappa_2
\]

Similar to lemma \ref{lemma1}, the asymptotic stability of this new system with the transformation, where the two states are:

\begin{equation}
\label{ts2}
\begin{aligned}
\dot{x}_7 &= z_2 + \phi(x_7) \\
\dot{z}_2 &= \varkappa_2
\end{aligned}
\end{equation}

can be determined using a composite Lyapunov function. 

\begin{proof}
Consider a positive definite and radially unbounded candidate defined by \ref{lyapunovC}
\[
V_{c2} = V_2 + \frac{1}{2} z_2^2
\]

\[
\dot{V}_{c2} = \dot{V}_2 + z_2 \dot{z}_2 = x_7z_2 -k_3 x_7^2 + z_2 (\varkappa_2)
\]

To ensure the asymptotic stability of this system, the following choice is made:

\[
\varkappa_2 = -x_7 - k_3 z_2
\]

Hence,

\[
\dot{V}_{c2} = -k_3 x_7^2 + z_2 (x_7 - x_7 - k_4 z_2) = -k_3 x_7^2 - k_4 z_2^2 < 0
\]

which shows asymptotic stability (AS) with \( K_3 \) and \( K_4 \) as controller gains, where \( K_3 \geq 0 \) and \( K_4 \geq 0 \).
\end{proof} 

Now, to find $U_{\phi}$, the following fact is used:

\( \varkappa_2 = \varsigma_2 - \dot{\phi}(x_7) \).

\begin{equation}
\begin{aligned}
\varsigma_2 &=\varkappa_2 + \dot{\phi}(x_7)
\end{aligned}
\end{equation}

For the system described by equation \eqref{mainsys2}, the control input \( U_\phi \) given in \eqref{control2} will stabilize the system:

\[
\frac{U_{\phi} - x_{10} x_{12} I_{yy} + x_{10} x_{12} I_{zz}}{I_{xx}} = -k_3 z_2 - x_7 - k_3
\]

\begin{equation}
\label{control2}
\begin{aligned}
U_{\phi} = I_{xx}(-k_3 z_2 - x_7 - k_3) + x_{10} x_{12} I_{yy} - x_{10} x_{12} I_{zz}
\end{aligned}
\end{equation}

A simple modification of the form,
\[
e_\phi = x_7 - \bar{x_\phi}
\]
where \( e_\phi \) is the error in the \(\phi\) orientation and \(\bar{x_\phi}\) is the desired orientation, will ensure the controller can be used for tracking purposes.

Following the same steps as above, the control inputs $U_\theta$ and $U_\psi$ are given below:

\begin{equation}
\label{control3}
\begin{aligned}
U_{\theta} = I_{yy}(-k_6 z_3 - x_9 - k_5) + x_{8} x_{12} I_{xx} + x_{8} x_{12} I_{zz}
\end{aligned}
\end{equation}

\begin{equation}
\label{control4}
\begin{aligned}
U_{\psi} = I_{zz}(-k_8 z_4 - x_{11} - k_7) - x_{8} x_{10} I_{xx} + x_{8} x_{10} I_{yy}
\end{aligned}
\end{equation}

\subsection{Ground controller design}

In this section, the backstepping controller is focused on the ground control input for the winder. The winder should be able to provide enough torque to maintain the required tension that will keep the tether in the catenary shape while following the desired trajectories of the TUAV. Similar to the backstepping steps followed earlier, the system of the winder will be considered in state space form as below:

\begin{equation}\label{mainsys3}
\begin{aligned}
\dot{x}_{13} &= x_{14} \\
\dot{x}_{14} &= \frac{1}{I_w} (- \beta_w x_{14} + r_w U_{\text{win}})
\end{aligned}
\end{equation}

A virtual controller is firstly designed: 

\begin{definition}
Let \( x_{14} \) be the input to this system, and define \( x_{14} = \phi(x_{13}) \) such that \( \phi(0) = 0 \). This choice of \( \phi \) ensures that the origin of the system is an equilibrium point and guarantees asymptotic stability (AS) and even global asymptotic stability.
\end{definition}

\begin{proposition}
The choice of \( \phi(x_{13}) = -k_{w} x_{13} \) is made.
\end{proposition}

\begin{proof}
Choosing a Lyapunov candidate defined as in \ref{lyapunovC}, \( V_{12} = \frac{1}{2} x_{13}^2 \).

\[
\dot{V}_{12} = x_{13} \dot{x}_{13} = x_{13} \cdot (-k_{w} x_{13}) = -k_{w} x_{13}^2 < 0 
\]

This shows global asymptotic stability. Hence, the chosen controller \( -k_{w} x_{13} \) will stabilize the system regardless of the initial condition.
\end{proof}

Now, consider a transformation, \( z_{7} = x_{14} - \phi(x_{13}) \).

\[
x_{14} = z_{7} + \phi(x_{13})
\]

The time derivative of this transformation is given by:

\[
\dot{z}_{7} = \dot{x}_{14} - \dot{\phi}(x_{13})
\]

Let \(\varsigma_7 = \dot{x}_8 \) and \(\varsigma_7 - \dot{\phi}(x_7) = \varkappa_7\)

Then,

\[
\dot{z}_7 = \varsigma_7 - \dot{\phi}(x_7) = \varkappa_7
\]

Also similar to Lemma \ref{lemma1}, the asymptotic stability of this new system can be proven with the transformation, where the two states are:
\[
\dot{x}_{13} = z_{7} + \phi(x_{13})
\]
\[
\dot{z}_{7} = \varkappa_7
\]

\begin{proof}
Similarly, the composite Lyapunov function
\[
V_{c12} = V_{12} + \frac{1}{2} z_{7}^2
\]

\[
\dot{V}_{c12} = \dot{V}_{12} + z_{7} \dot{z}_{7} =  x_{13}z_7 -k_{w} x_{13}^2 + z_{7} (x_{13} + \varkappa_7)
\]

To ensure the asymptotic stability of this system, the following choice is made:

\[
\varkappa_7 = -x_{13} - k_{w2} z_{7}
\]

Hence,

\begin{equation*}
\begin{aligned}
\dot{V}_{c12} &= -k_{w} x_{13}^2 + z_{7} (x_{13} - x_{13} - k_{w2} z_{7}) \\
              &= -k_{w} x_{13}^2 - k_{w2} z_{7}^2 < 0
\end{aligned}
\end{equation*}

which shows asymptotic stability (AS) with \( K_w \) and \( K_{w2} \) as controller gains, where \( K_w \geq 0 \) and \( K_{w2} \geq 0 \).
\end{proof}

Now, to find \( U_{\text{win}} \), 

\( u_{\text{win}} = \varsigma_7 - \dot{\phi}(x_{13}) \).

\begin{equation}
\begin{aligned}
\varsigma_7 &= \varkappa_7  + \dot{\phi}(x_{13})
\end{aligned}
\end{equation}

For the system described by equation \eqref{mainsys3}, the control input \( U_win \) given in \eqref{control5} will stabilize the system:

\[
\frac{1}{I_w} (-\beta_w + r_w U_{\text{win}}) = -x_{13} - k_{w2} z_{7} - k_{w} x_{13}
\]

\begin{equation}
\label{control5}
\begin{aligned}
U_{\text{win}} = \frac{I_w (-x_{13} - k_{w2} z_{7} - k_{w} x_{13}) + \beta_2 x_{14}}{r_w}
\end{aligned}
\end{equation}

A simple modification of the form,
\[
e_{L} = x_{13} - \bar{L}
\]
where \( e_{L} \) is the error in the tether length and \( \bar{L} \) is the desired length based on the TUAV current position, will ensure the control input can be used to track desired trajectories of the tether.

\section{Results and Simulation}
Consider the TUAV system described in Section \ref{sysdes}, with the following parameters: mass \( m = 2.84 \) kg, gravitational acceleration \( g = 9.81 \) m/s\(^2\), and moments of inertia \( I_{xx} = 0.5192 \) kg\(\cdot\)m\(^2\), \( I_{yy} = 0.4929 \) kg\(\cdot\)m\(^2\), and \( I_{zz} = 0.0947 \) kg\(\cdot\)m\(^2\). The density of the tether is given as \( \rho = 0.034 \) kg/m, with a cross-sectional area of \( A = 1.1 \times 10^{-4} \) m\(^2\). In this section, the simulation results of the developed controllers using these parameters are presented.

\subsection{Setpoint and Errors}
First, the state evolution through setpoint tracking is demonstrated, illustrating that all states asymptotically converge to zero, as shown in Fig. \ref{states_evolve}.

\begin{figure}[ht]
\centering
\includegraphics[width=0.5\textwidth]{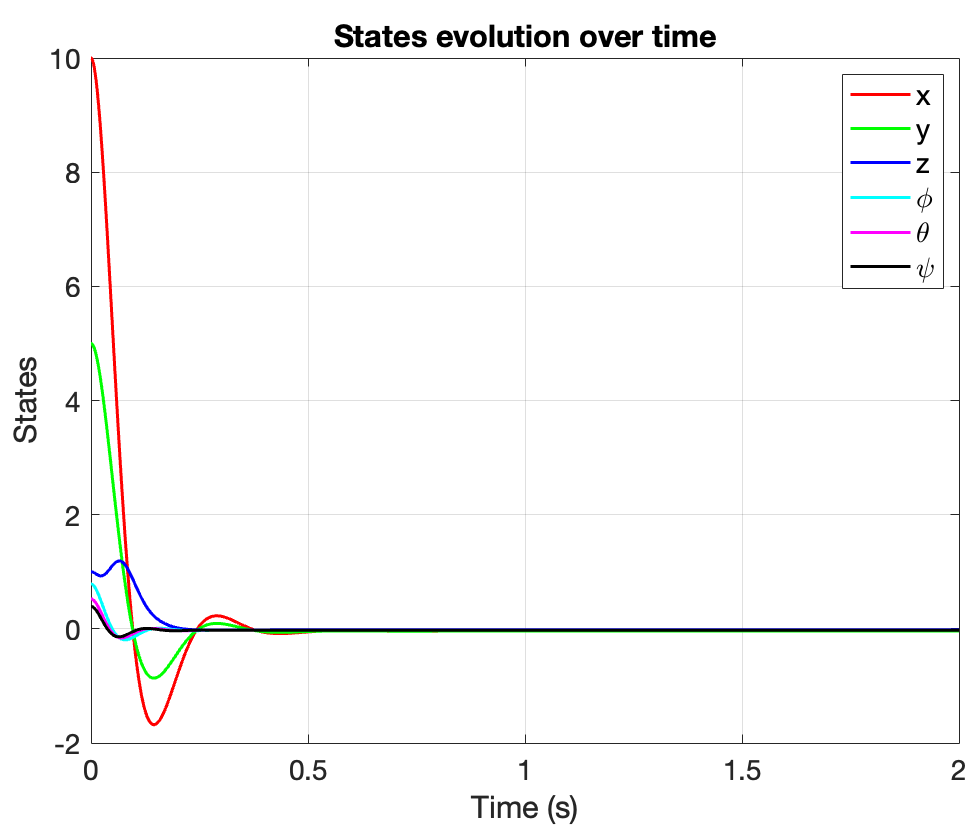}
\caption{Stabilization of system states over time}
\label{states_evolve}
\end{figure}

As the states evolve over time, the error in the states is expected to follow similar trajectories and converge to zero in finite time, as depicted in Fig. \ref{errors}.

\begin{figure}[ht]
\centering
\includegraphics[width=0.5\textwidth]{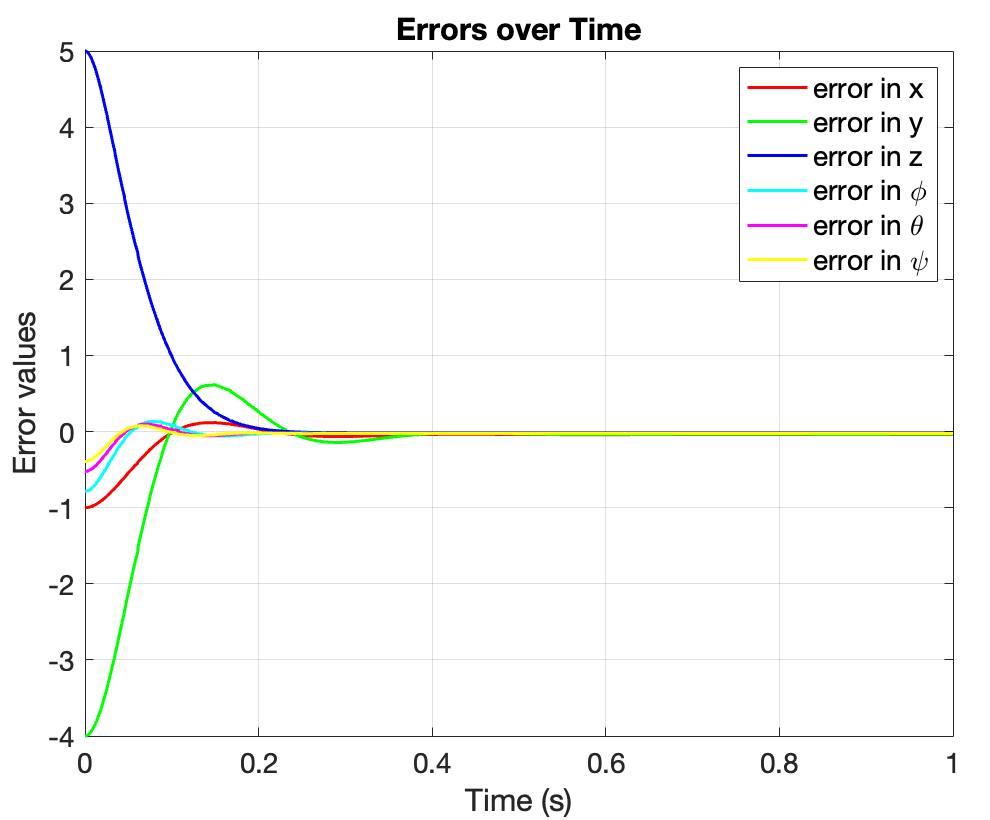}
\caption{Errors over time}
\label{errors}
\end{figure}

\subsection{Linear Trajectory Tracking}

Following the demonstration that all states globally asymptotically stabilize to the setpoint in finite time, the performance of the control under different operational conditions was evaluated. First, trajectory tracking of the six states was tested. In Fig. \ref{pos}, it can be seen that the designed controller was able to track the desired trajectories $(\bar{x}, \bar{y}, \bar{z}) \in \mathbb{R}^3$ of the system, and Fig. \ref{angles} shows similar tracking of the UAV's attitude.

\begin{figure}[ht]
\centering
\includegraphics[width=0.5\textwidth]{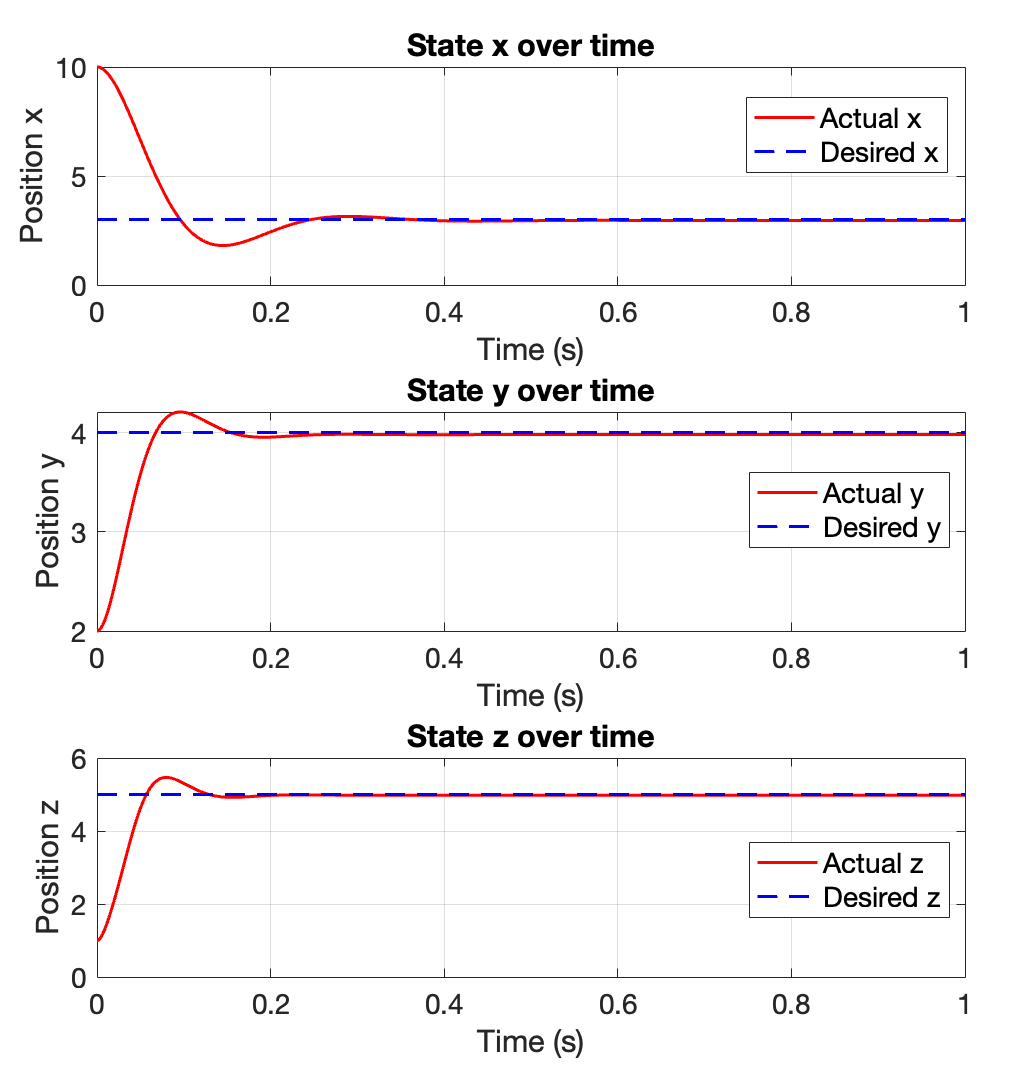}
\caption{Position tracking of TUAV}
\label{pos}
\end{figure}

\begin{figure}[ht]
\centering
\includegraphics[width=0.5\textwidth]{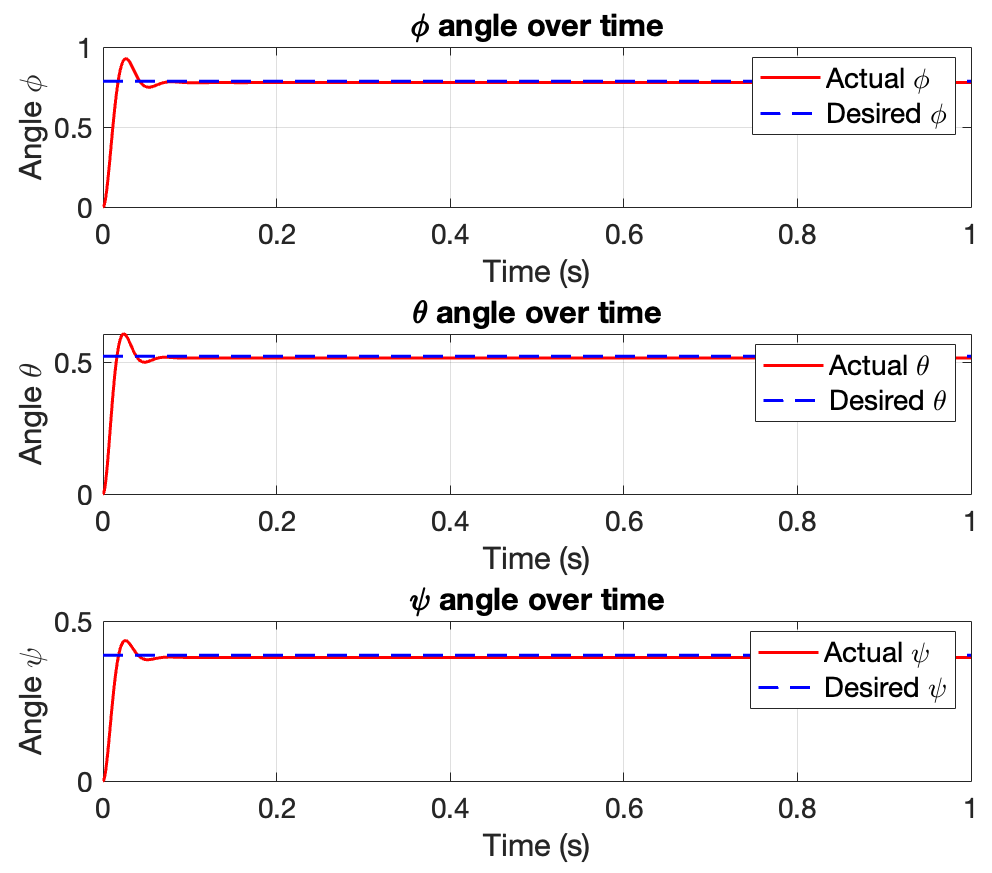}
\caption{Attitude tracking of TUAV}
\label{angles}
\end{figure}

\subsection{Circular Trajectory Tracking}
Following the tracking of linear trajectories for the system's position and attitude, the controller's performance was tested on more complex trajectories, such as a circular one. This trajectory is significant for TUAV operations, as the stationary nature of the system necessitates the ability to ascend, maintain a position, and maneuver at a specific altitude. This capability is particularly important in applications involving agriculture and surveillance. The results are shown in Fig. \ref{circular}.

\begin{figure}[ht]
\centering
\includegraphics[width=0.5\textwidth]{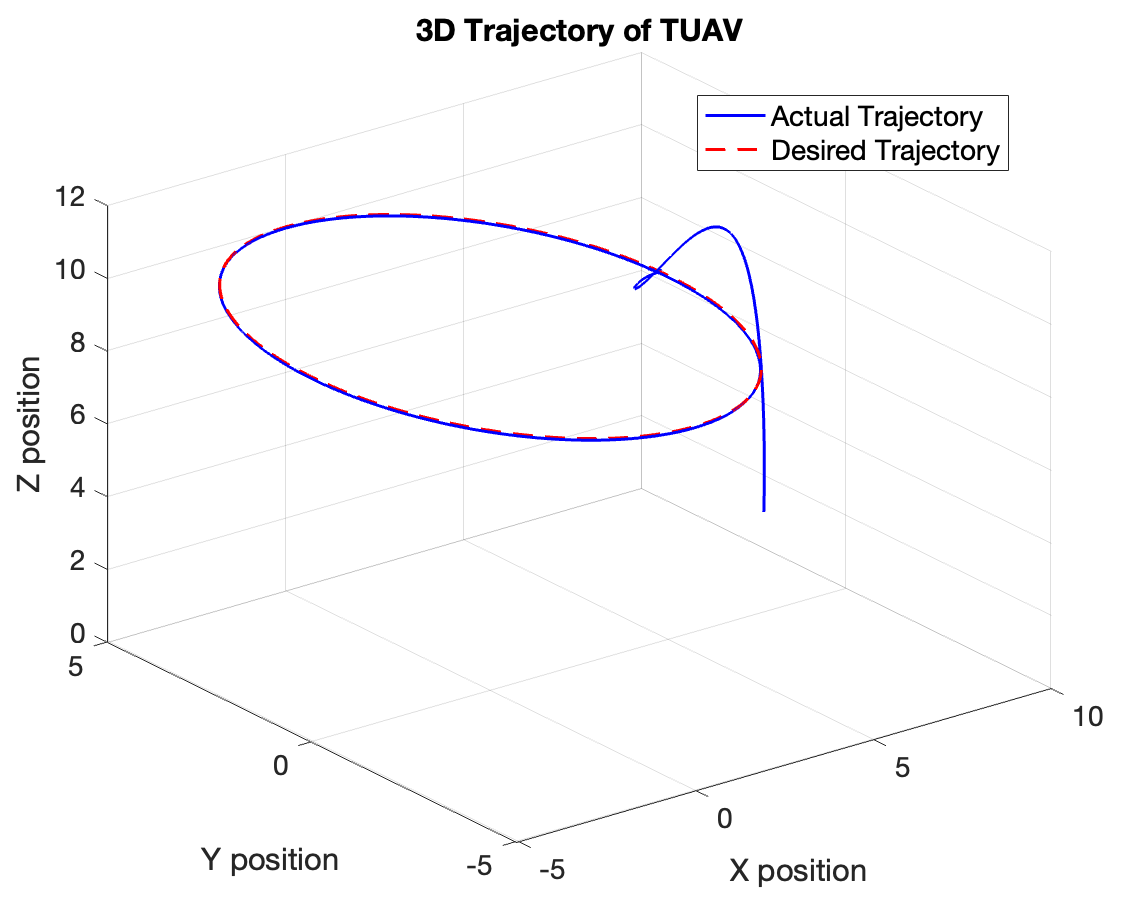}
\caption{UAV tracking a circular trajectory}
\label{circular}
\end{figure}

\subsection{Tether Length Tracking}

Having demonstrated that the TUAV's position and attitude follow the desired trajectories, it was necessary to ensure that the tether also follows the system while maintaining its catenary nature. In Fig. \ref{tether}, the evolution of the tether is shown by comparing the desired and actual states as the TUAV moves from its initial position to its desired position.

\begin{figure}[ht]
\centering
\includegraphics[width=0.5\textwidth]{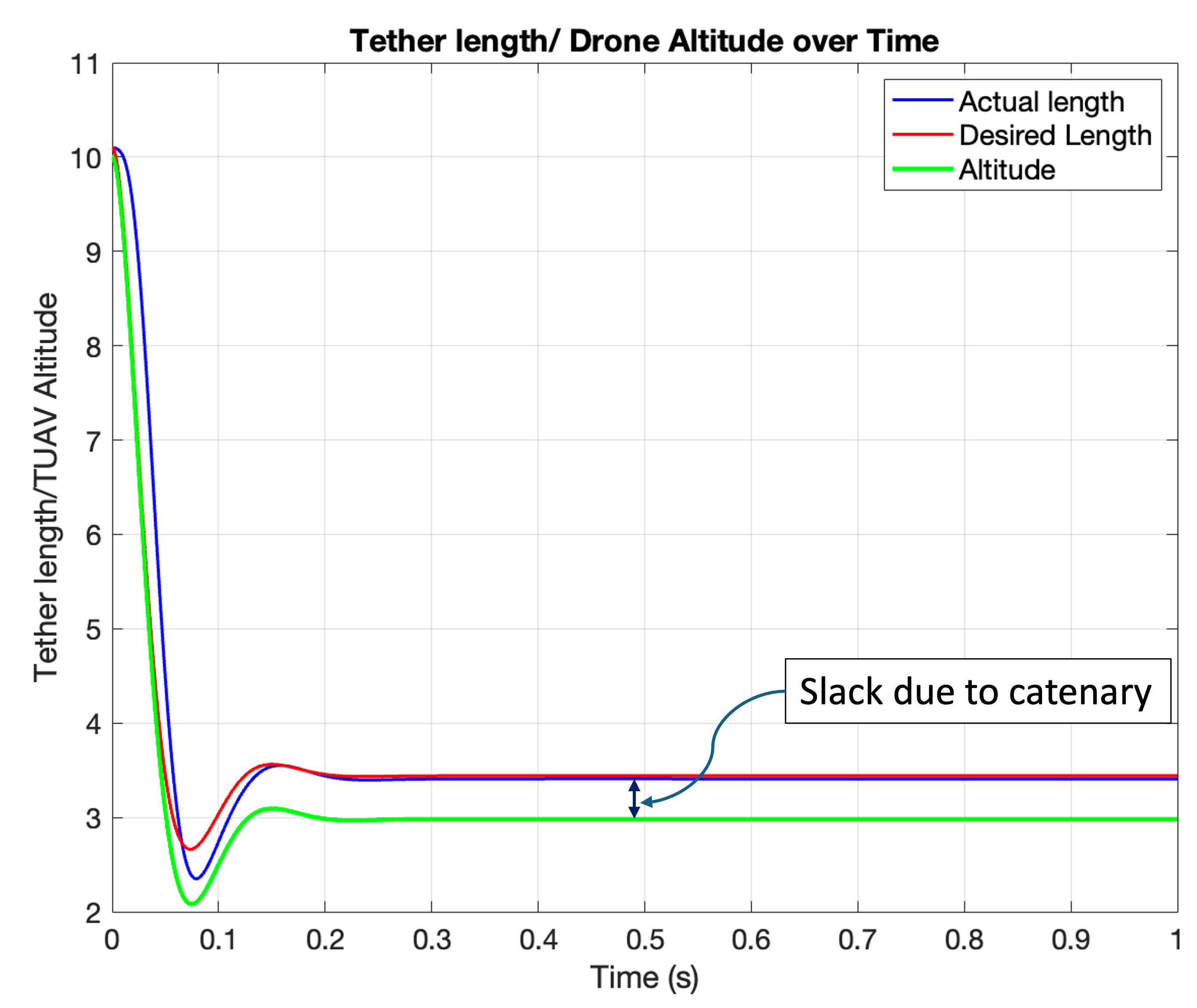}
\caption{Tether length tracking}
\label{tether}
\end{figure}

The altitude of the system is also included to show that the tether length adjusts in accordance with this state of the system. It should be noted that, due to the catenary nature of the tether, the desired and actual length of the tether are always greater than the altitude of the TUAV due to the slack nature of the catenary. The catenary length of the tether is in accordance with equation \ref{eqn2}.

\subsection{Animations}
To further demonstrate the effectiveness of the designed controller and the catenary nature of the tether, animations of the system were created. Readers can check the project github page at \href{https://github.com/sof-danny/TUAV\_system\_control}{https://github.com/sof-danny/TUAV\_system\_control}
 for these animations. In the animations, readers can observe the ripple effect of wind on the tether as the TUAV moves from its initial position to the desired location. Both single position tracking and multiple position tracking were tested. Since only snapshots can be shown in this paper, the trail of the TUAV's movement has been added to illustrate how the system transitions from point to point.

\subsubsection{Single Position}
Here, the goal is to move from the initial condition to a single desired position while maintaining the catenary shape of the tether. As seen in a snapshot in Fig. \ref{single_animation}, the TUAV was able to reach the desired location and stay there. This is useful for applications involving surveillance where one wants the TUAV to cover a particular section of the region of interest.

\begin{figure}[ht]
\centering
\includegraphics[width=0.5\textwidth]{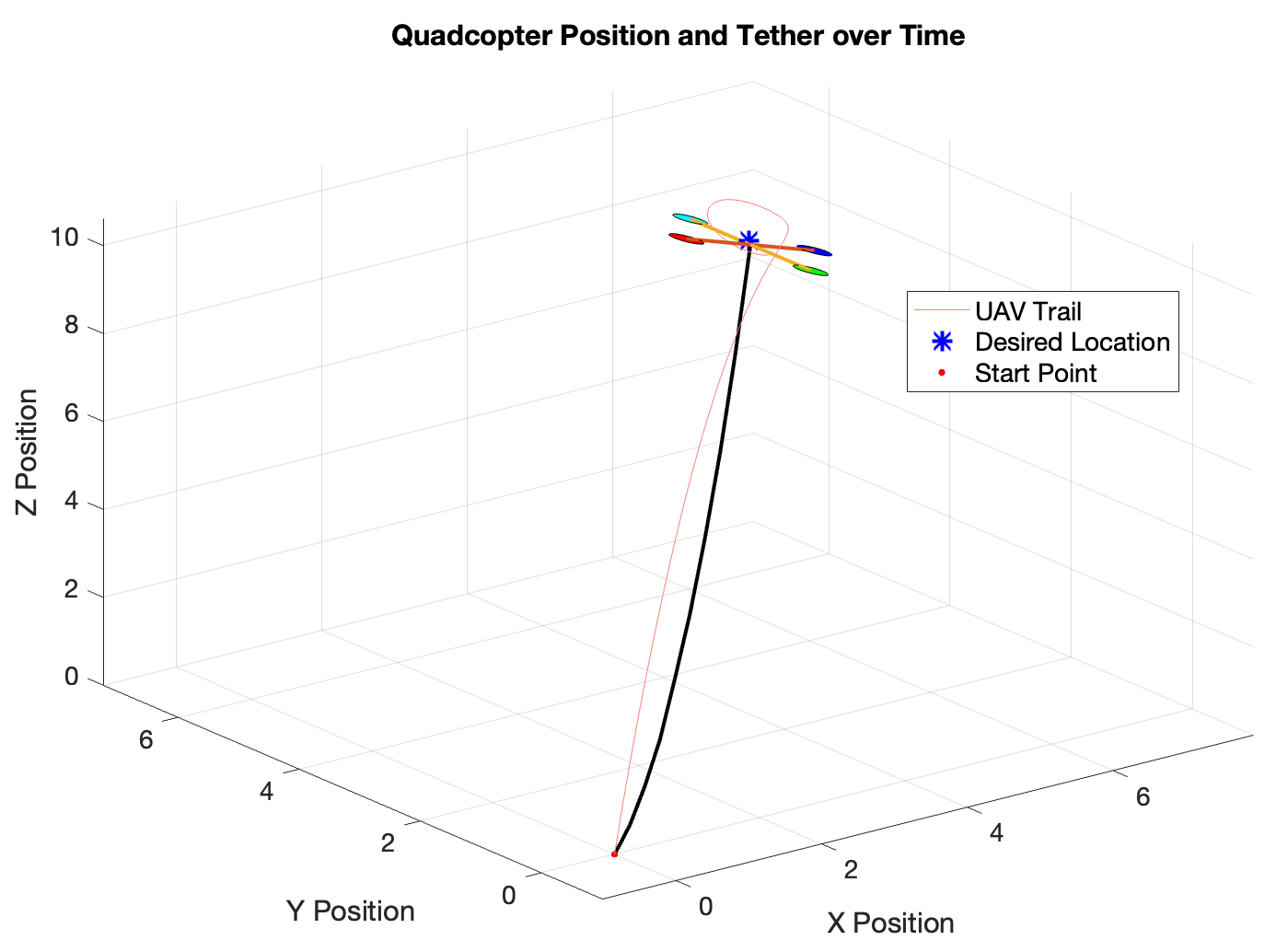}
\caption{TUAV reaching a desired location}
\label{single_animation}
\end{figure}

\subsubsection{Multiple Positions}
In many other instances, such as agriculture, dynamic surveillance, photography, and construction, one may want the TUAV to move from one position to another without returning to the starting position. In Fig. \ref{multiple_animation}, the ability of the system to track multiple positions is shown. The initial conditions are reinitialized after the TUAV reaches the first position, making the first position the next initial position.

\begin{figure}[ht]
\centering
\includegraphics[width=0.5\textwidth]{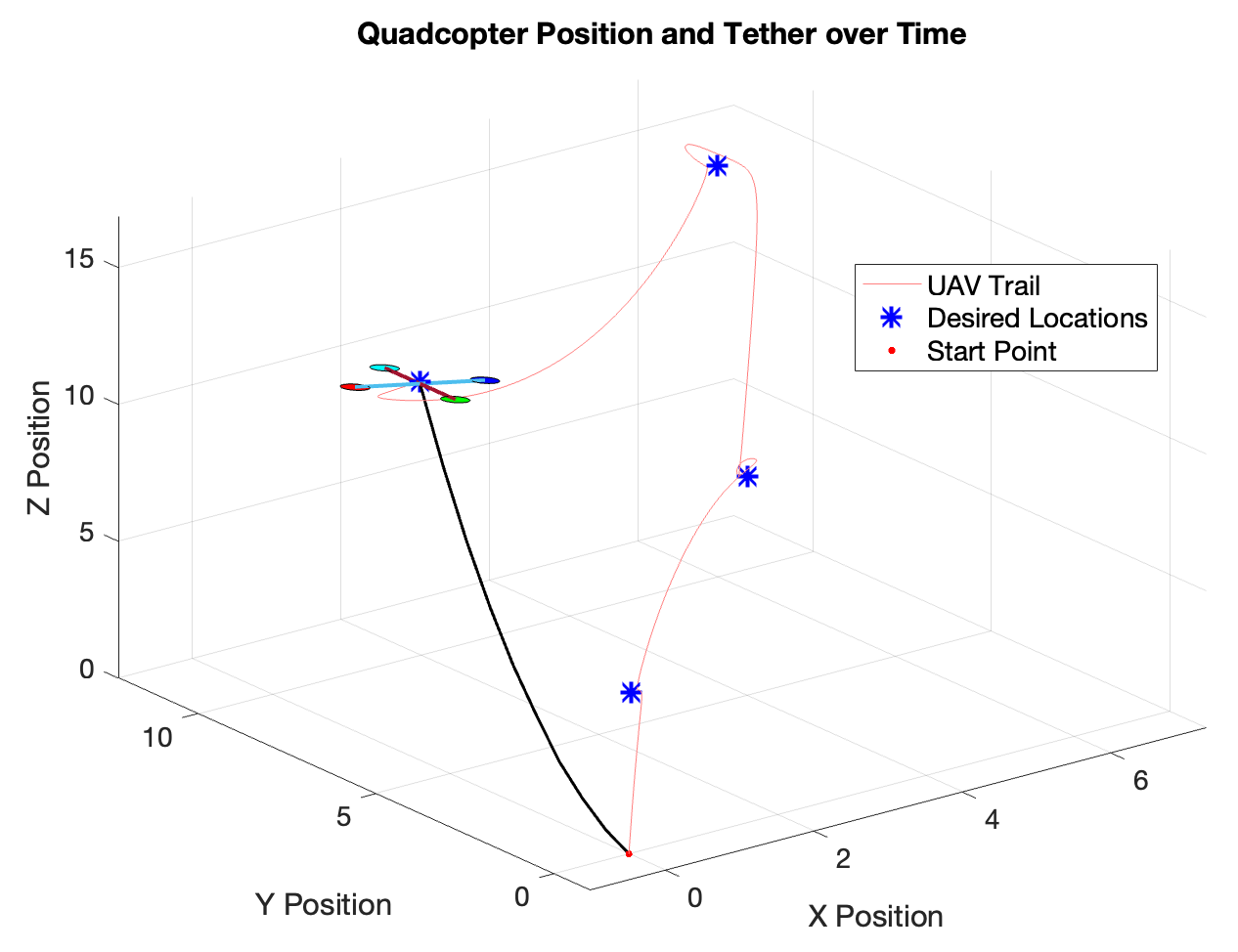}
\caption{TUAV reaching multiple locations}
\label{multiple_animation}
\end{figure}

\section{Conclusion}
In this paper, a comprehensive control architecture for a Tethered Unmanned Aerial Vehicle (TUAV) system, integrating both onboard and ground-based controllers was developed. The onboard controllers were tasked with managing the position and attitude of the UAV, while the ground control was responsible for regulating the winder to maintain the catenary shape of the tether. 
Utilizing nonlinear backstepping control techniques, global asymptotic stability of the UAV's states was ensured.

The effectiveness of the control strategy was validated through extensive simulations. The results demonstrated the TUAV system's ability to accurately track both linear and circular trajectories, highlighting the robustness of the control design in various operational scenarios. Additionally, it was ensured that the tether length adapted appropriately to maintain the catenary shape, which is critical for the system's stability and functionality.

Further, animations were shown to illustrate the dynamic performance of the TUAV system, providing visual evidence of the system's capability to handle both single and multiple position tracking tasks. These animations reinforced the practical applicability of the control approach in real-world scenarios, such as surveillance, agriculture, construction, and other field operations.

The next phase of this project could involve integrating the dynamics of mobile platforms, such as a cart or vehicle, to enhance the TUAV system's mobility and ensure it is not constrained to a fixed location.


\bibliographystyle{IEEEtran}

\bibliography{IEEEfull}

\begin{IEEEbiography}[{\includegraphics[width=1in,height=1.25in,clip,keepaspectratio]{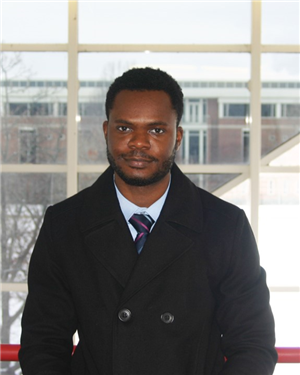}}]{SAMUEL O. FOLORUNSHO is a Graduate member of IEEE and
a PhD student conducting research at the Autonomous and
Unmanned Vehicle Systems Lab (AUVSL) and the Center for Autonomous
Construction and Manufacturing at Scale (CACMS) in the Department of
Systems Engineering at the University of Illinois, Urbana-Champaign (UIUC) under the advisorship of Prof. William R. Norris. His research is focused on control systems, computer vision and robotics - and the intersection of those for
safety-critical systems in industrial and agricultural applications.
He earned his M.S. from UIUC in 2023 and his B.S. in 2017 at
the University of Ilorin, Nigeria both in Agricultural and Biological
Systems Engineering. He has three years of working experience in management consulting.}
\end{IEEEbiography}

\vspace{5pt}

\begin{IEEEbiography}[{\includegraphics[width=1in,height=1.25in,clip,keepaspectratio]{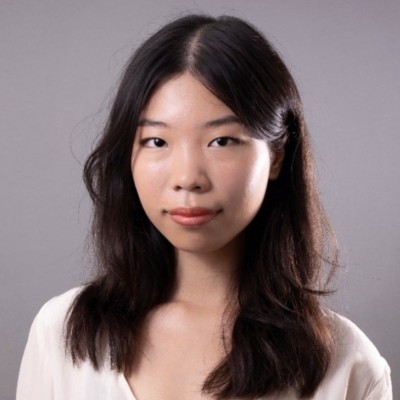}}]{MAGGIE NI is currently pursuing the B.S. degree in Aerospace Engineering at the University of Illinois Urbana-Champaign (UIUC). She is an undergraduate research assistant at the Autonomous and Unmanned Vehicle System Laboratory (AUVSL), founded by Prof. William R. Norris. Her research interests include autonomous systems, UAV navigation, and control systems}
\end{IEEEbiography}

\vspace{5pt}

\begin{IEEEbiography}[{\includegraphics[width=1in,height=1.25in,clip,keepaspectratio]{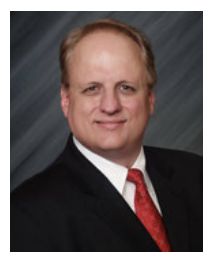}}]{WILLIAM R. NORRIS (Member, IEEE) received
the B.S., M.S., and Ph.D. degrees in systems engineering from the University of Illinois at Urbana–
Champaign, in 1996, 1997, and 2001, respectively,
and the M.B.A. degree from the Fuqua School of
Business, Duke University, in 2007. He has over
23 years of industry experience with autonomous
systems. He is currently a Research Professor with the Industrial and Enterprise Systems
Engineering Department, University of Illinois at
Urbana–Champaign, the Director of the Autonomous and Unmanned Vehicle
System Laboratory (AUVSL), as well as the Founding Director of the Center
for Autonomous Construction and Manufacturing at Scale (CACMS).}
\end{IEEEbiography}

\vfill

\end{document}